\documentclass[12pt,a4paper]{article}%
\usepackage{amsmath}
\usepackage{amsfonts}
\usepackage{amssymb}
\usepackage{graphicx}%
\newtheorem{theorem}{Theorem}

\newtheorem{corollary}[theorem]{Corollary}

\newtheorem{definition}[theorem]{Definition}

\newtheorem{lemma}[theorem]{Lemma}

\newtheorem{remark}[theorem]{Remark}

\newenvironment{proof}[1][Proof]{\noindent\textbf{#1.} }{\ \rule{0.5em}{0.5em}}
\begin{document}
\begin{titlepage}
\begin{flushright}
DISTA-2010 \\
\par\end{flushright}
\vskip 1.5cm
\begin{center}
\textbf{\LARGE \v Cech and de Rham Cohomology}
\textbf{ }\\
\textbf{ }\\
\textbf{\LARGE of Integral Forms}
\textbf{\vfill{}}
\textbf{\large R.Catenacci$^\ddag$, M.Debernardi, P.A.Grassi$^\ddag$, D.Matessi}
{\vfill{}
{{ DISTA, Universit\`{a} del Piemonte Orientale,
}}\\
{{ Via Teresa Michel, 11, Alessandria, 15120, Italy }}\\
{$^\ddag${INFN
- Sezione di Torino - Gruppo Collegato di Alessandria}}\\
\textbf{} }
\par\end{center}
\vfill{}
\begin{abstract}
We present a study on the integral forms and their \v Cech and de Rham cohomology.
We analyze the problem from a general perspective of sheaf theory and we explore examples
in superprojective manifolds. Integral forms are fundamental in the theory of integration in supermanifold.
One can define the integral forms introducing a new sheaf containing, among other objects,
the new basic forms $\delta(d\theta)$ where the symbol $\delta$ has the usual formal properties of
Dirac's delta distribution and acts on functions and forms as a Dirac measure. They satisfy in addition some
new relations on the sheaf. It turns out that  the enlarged sheaf of integral and "ordinary"
superforms contains also forms of "negative degree" and, moreover, due to the additional relations introduced it
is, in a non trivial way, different from the usual superform cohomology.
\end{abstract}
\vfill{}
\vspace{1.5cm}
\end{titlepage}

\section{Introduction}

Supermanifolds are rather well-known in supersymmetric theories and in string
theory. They provide a very natural ground to understand the supersymmetry and
supergravity from a geometric point of view. Indeed, a supermanifold contains
the anticommuting coordinates which are needed to construct the superfields
whose natural environment is the graded algebras \cite{Wess:1992cp,manin}.
However, the best way to understand the supermanifold is using the theory of
sheaves \cite{manin,Bruzzo}.

Before explaining the content of the present work, we stress the relevance of this 
analysis observing that recently the construction of a formulation of superstrings 
\cite{berko}
requires the introduction of the superforms described here. In addition, the 
physics behind that formalism is encoded into the BRST cohomology which, 
in mathematical terms, is translated into the \v Cech and de Rham cohomology 
objects of our study. 

In the present notes we review this approach and we use the results of our
previous paper \cite{us}. In the first section, we recall some definitions and
some auxiliary material. We point out that in order to formulate the theory of
integration for superforms, one needs some additional ingredients such as
integral forms. Enlarging the space of superforms to take into account those
new quantities results in bigger complexes of superforms. These new
mathematical objects should be understood in the language of the sheaves in
order that the previous considerations about the morphisms are applicable. In
particular, we study the behaviour of integral forms under morphisms and we
show that they can be globally defined. 

By a hand-waving argument, we can describe as follows the need of the integral
forms for the theory of integration in the supermanifold. In the theory of
integration of conventional forms for a manifold $\mathcal{M}$, we consider a
$\omega\in\Omega^{\bullet}(\mathcal{M})$. We can introduce a supermanifold
\cite{Cordes:1994fc} $\widehat{\mathcal{M}}$ whose anticommuting coordinates
are generated by the fibers $T^{\ast}\mathcal{M}$. Therefore, a function on
$\widehat{\mathcal{M}}$ is the same of a differential form of $\Omega
^{\bullet}(\mathcal{M})$, $\mathcal{F}(\widehat{\mathcal{M}})\equiv
\mathcal{C}^{\infty}(\widehat{\mathcal{M}})\cong\Omega^{\bullet}(\mathcal{M}%
)$. The correspondence is simply $dx^{i}\leftrightarrow\psi^{i}$. For the
manifold $\mathcal{M}$ we can integrate differential forms of the top degree
$\Omega^{(n)}(\mathcal{M})$, but in general we cannot integrate functions
since $\mathcal{M}$ has no natural measure. On the other hand in
$\widehat{\mathcal{M}}$ we can indeed write $\widehat{\mu}=dx^{1}%
\wedge\,...\,\wedge dx^{n}\wedge d\psi^{1}\wedge\,...\,\wedge d\psi^{n}$ where
the integral on the variables $\psi^{i}$ is the Berezin integral ($\int d\psi f(\psi)=\partial_{\psi}f(\psi)$). If $\hat{\omega
}$ is a function of $\mathcal{F}(\widehat{\mathcal{M}})$, we have
$\int_{\widehat{\mathcal{M}}}\widehat{\mu}\,\widehat{\omega}=\int
_{\mathcal{M}}\omega$ where the superspace integration is the integration of
forms. We have to notice that being the integral on the anticommuting variables
a Berezin integral, it selects automatically the degree of the form.

Now, the same construction can be performed in the case of a supermanifold
$\mathcal{N}$ with only fermionic coordinates $\theta^{a}$. In that case its
cotangent space $T^{\ast}\mathcal{N}$ is not finite dimensional. Therefore,
mimicking the above construction, we define the form integration by
considering the measure $\widehat{\mu}$ for the manifold $\mathcal{N}\oplus
T^{\ast}\mathcal{N}$ where the commuting superforms $d\theta^{a}$ are replaced
by commuting coordinates $\lambda^{a}$ and the measure is given $\widehat{\mu
}=d\theta^{a}\wedge\,...\,d\theta^{b}\wedge d\lambda^{a}\wedge\,...\,d\lambda
^{b}$. Thus, in contrapposition to the commuting case the integral over the
coordinates $\theta^{a}$ is a Berezin integral, while the integration over the
1-forms $\lambda^{a}$ is an ordinary $n$-dimensional integral. In order that
the latter has finite answer for a given superform, we introduce the
integration forms $\omega_{a_{1}\dots a_{n}}\delta(\lambda^{a_{1}}%
)\wedge\,...\,\wedge\delta(\lambda^{a_{n}})$ where the Dirac delta functions
$\delta(\lambda^{a})$ localize the integral at the point $\lambda^{a}=0$.
These new quantities behave as ``distributions'', and therefore they satisfy new relations that we will describe in Sec. 4. We show that the set of relation they 
ought to obey are preserved in passing from one patch
to another and therefore that they are global properties. This implies that the 
sheaf of integral forms is well defined. Finally, we derive a
\v{C}ech - de Rham theorem for these new superforms. The interesting aspect is
that the distributional relations (here translated into an algbebraic
language) modifies the cohomology and therefore we find non-trivial results.

In sec. 2, we review briefly the construction of the supermanifolds, the
underlying structure using ringed spaces, their morphisms and the local charts
on them. We specify the constructions to the case of superprojective manifolds.
In sec. 5 and in sec. 6 we compute some examples of \v{C}ech and de Rham cohomology groups for superprojective spaces. We also prove a generalization of usual \v{C}ech-de Rham and K\"unneth theorems. 

\section{Supermanifolds}

\label{defs}

We collect here some definitions and considerations about supermanifolds

\subsection{Definitions}

A \textbf{super-commutative} ring is a $\mathbb{Z}_{2}$-graded ring $A = A_{0}
\oplus A_{1}$ such that if $i,j \in\mathbb{Z}_{2}$, then $a_{i} a_{j} \in
A_{i+j}$ and $a_{i} a_{j} = (-1)^{i+j}a_{j}a_{i}$, where $a_{k} \in A_{k} $.
Elements in $A_{0}$ (resp. $A_{1}$) are called \textbf{even} (resp.
\textbf{odd}).

A \textbf{super-space} is a super-ringed space such that the stalks are local
super-commutative rings (Manin-Varadarajan). Since the odd elements are
nilpotent, this reduces to require that the even component reduces to a local
commutative ring.

A \textbf{super-domain} $U^{p|q}$ is the super-ringed space $\left(
U^{p},\mathcal{C}^{\infty p|q}\right)  $, where $U^{p}\subseteq\mathbb{R}^{p}$
is open and $\mathcal{C}^{\infty p|q}$ is the sheaf of super-commutative rings
given by:
\begin{equation}
V \mapsto\mathcal{C}^{\infty}\left(  V\right)  \left[  \theta^{1},\theta
^{2},...,\theta^{q}\right]  ,
\end{equation}
where $V \subseteq U^{p}$ and $\theta^{1},\theta^{2},...,\theta^{q} $ are
generators of a Grassmann algebra. The grading is the natural grading in even
and odd elements. The notation is taken from \cite{Varadarajan:2004yz} and
from the notes \cite{Varadarajan:note}.

Every element of $\mathcal{C}^{\infty p|q}\left(  V\right)  $ may be written
as $\sum_{I}f_{I}\theta^{I}$, where $I$ is a multi-index. A
\textbf{super-manifold} of dimension $p|q$ is a super-ringed space locally
isomorphic, as a ringed space, to $\mathbb{R}^{p|q}$. The coordinates $x_{i}$
of $\mathbb{R}^{p}$ are usually called the even coordinates (or bosonic),
while the coordinates $\theta^{j}$ are called the odd coordinates (or
fermionic). We will denote by $\left(  M,\mathcal{O}_{M}\right)  $ the
supermanifold whose underlying topological space is $M$ and whose sheaf of
super-commutative rings is $\mathcal{O}_{M}$.

To a section $s$ of $\mathcal{O}_{M}$ on an open set containing $x$ one may
associate the \textbf{value} of $s$ in $x$ as the unique real number
$\tilde{s}\left(  x\right)  $ such that $s-\tilde{s}\left(  x\right)  $ is not
invertible on every neighborhood of $x$. The sheaf of algebras $\overset{\sim
}{\mathcal{O}}$, whose sections are the functions $\tilde{s}$, defines the
structure of a differentiable manifold on $M$, called the \textbf{reduced
manifold} and denoted $\overset{\sim}{M}$.

\subsection{Morphisms.}

In order to understand the structure of supermanifolds it is useful to study
their morphisms. Here we describe how a morphism of supermanifolds looks like
locally. A \textbf{morphism} $\psi$ from $(X,\mathcal{O}_{X})$ to
$(Y,\mathcal{O}_{Y})$ is given by a smooth map $\overset{\sim}{\psi}$ from
$\overset{\sim}{X}$ to $\overset{\sim}{Y}$ together with a sheaf map:
\begin{equation}
\psi_{V}^{\ast}:\mathcal{O}_{Y}(V)\longrightarrow\mathcal{O}_{X}(\psi
^{-1}(V)),
\end{equation}
where $V$ is open in $Y$. The homomorphisms $\psi_{V}^{\ast}$ must commute
with the restrictions and they must be compatible with the super-ring
structure. Moreover they satisfy
\begin{equation}
\psi_{V}^{\ast}(s)^{\sim}=\tilde{s}\circ\tilde{\psi}.
\end{equation}
Let us recall some fundamental local properties of morphisms. A morphism
$\psi$ between two super-domains $U^{p|q}$ and $V^{r|s}$ is given by a smooth
map $\tilde{\psi}:U\rightarrow V$ and a homomorphism of super-algebras
\begin{equation}
\psi^{\ast}:\mathcal{C}^{\infty\,r|s}(V)\rightarrow\mathcal{C}^{\infty
\,p|q}(U).
\end{equation}
It must satisfy the following properties:
\begin{itemize}
\item If $t=(t_{1},\ldots,t_{r})$ are coordinates on $V^{r}$, each component
$t_{j}$ can also be interpreted as a section of $\mathcal{C}^{\infty\,
r|s}(V)$. If $f_{i}=\psi^{\ast}(t_{i})$, then $f_{i}$ is an even element of
the algebra $\mathcal{C}^{\infty\, p|q}(U)$.

\item The smooth map $\tilde{\psi}:U\rightarrow V$ must be $\tilde{\psi
}=(\tilde{f}_{1},\ldots,\tilde{f}_{r})$, where the $\tilde{f}_{r}$ are the
values of the even elements above.

\item If $\theta_{j}$ is a generator of $\mathcal{C}^{\infty\,r|s}(V)$, then
$g_{j}=\psi^{\ast}(\theta_{j})$ is an odd element of the algebra
$\mathcal{C}^{\infty\,p|q}(U)$.
\end{itemize}

The following fundamental theorem (see for example \cite{Varadarajan:note})
gives a local characterization of morphisms:

\begin{theorem}
{[}\textbf{Structure of morphisms}] \label{morphisms} Suppose $\phi
:U\rightarrow V$ is a smooth map and $f_{i},g_{j}$, with $i=1,\ldots,r$,
$j=1,\ldots,s$, are given elements of $\mathcal{C}^{\infty\,p|q}(U)$, with
$f_{i}$ even, $g_{j}$ odd and satisfying $\phi=(\tilde{f}_{1},\ldots,\tilde
{f}_{r})$. Then there exists a unique morphism $\psi:U^{p|q}\rightarrow
V^{r|s} $ with $\tilde{\psi}=\phi$ and $\psi^{\ast}(t_{i})=f_{i}$ and
$\psi^{\ast}(\theta_{j})=g_{j}$.
\end{theorem}

\subsection{Local charts on supermanifolds}

\label{smfld} We describe now how supermanifolds can be constructed by
patching local charts. Let $X=\bigcup_{i}X_{i}$ be a topological space, with
$\{X_{i}\}$ open, and let $\mathcal{O}_{i}$ be a sheaf of rings on $X_{i}$,
for each $i$. We write (see \cite{Varadarajan:2004yz}) $X_{ij}=X_{i}\cap
X_{j}$, $X_{ijk}=X_{i}\cap X_{j}\cap X_{k}$, and so on. We now introduce
isomorphisms of sheaves which represent the \textquotedblleft coordinate
changes\textquotedblright\ on our super-manifold. They allow us to glue the
single pieces to get the final supermanifold. Let
\begin{equation}
f_{ij}:\left(  X_{ji},\mathcal{O}_{j}|_{X_{ji}}\right)  \longrightarrow\left(
X_{ij},\mathcal{O}_{i}|_{X_{ij}}\right)
\end{equation}
be an isomorphisms of sheaves with
\begin{equation}
\tilde{f}_{ij}=Id.
\end{equation}
This means that these maps represent differentiable coordinate changes on the
underlying manifold.

To say that we glue the ringed spaces $(X_{i},\mathcal{O}_{i})$ through the
$f_{ij}$ means that we are constructing a sheaf of rings $\mathcal{O}$ on $X$
and for each $i$ a sheaf isomorphism
\begin{equation}
f_{i}:(X_{i},\mathcal{O}|_{X_{i}})\longrightarrow(X_{i},\mathcal{O}_{i}),
\end{equation}%
\begin{equation}
\tilde{f}_{i}=Id_{X_{i}}%
\end{equation}
such that
\begin{equation}
f_{ij}=f_{i}f_{j}^{-1},
\end{equation}
for all $i$ and $j$.

The following usual cocycle conditions are necessary and sufficient for the
existence of the sheaf $\mathcal{O}$:

\begin{description}
\item[i.] $f_{ii}=Id$ on $\mathcal{O}_{i}$;

\item[ii.] $f_{ij}f_{ji}=Id$ on $\mathcal{O}_{i}|_{X_{i}}$;

\item[iii.] $f_{ij}f_{jk}f_{ki}=Id$ on $\mathcal{O}_{i}|_{X_{ijk}}$.
\end{description}

\section{Projective superspaces}

\label{proj} Due to their importance in mathematical and physical applications
we now give a description of projective superspaces (see also \cite{us}%
). One can work either on $\mathbb{R}$ or on $\mathbb{C}$, but we choose to
stay on $\mathbb{C}$. Let $X$ be the complex projective space of dimension
$n$. The super-projective space will be called $Y$. The homogeneous
coordinates are $\left\{  z_{i}\right\}  $. Let us consider the underlying
topological space as $X$, and let us construct the sheaf of super-commutative
rings on it. For any open subset $V\subseteq X$ we denote by $V^{\prime}$ its
preimage in $\mathbb{C}^{n+1}\setminus\left\{  0\right\}  $. Then, let us
define $A\left(  V^{\prime}\right)  =H\left(  V^{\prime}\right)  \left[
\theta^{1},\theta^{2},...,\theta^{q}\right]  $, where $H\left(  V^{\prime
}\right)  $ is the algebra of holomorphic functions on $V^{\prime}$ and
$\left\{  \theta^{1},\theta^{2},...,\theta^{q}\right\}  $ are the odd
generators of a Grassmann algebra. $\mathbb{C}^{\ast}$ acts on this
super-algebra by:
\begin{equation}
t:{\sum_{I}}f_{I}\left(  z\right)  \theta^{I}\longrightarrow{\sum_{I}}%
t^{-|I|}f_{I}\left(  t^{-1}z\right)  \theta^{I}. \label{nonzero}%
\end{equation}
The super-projective space has a ring over $V$ given by:
\begin{equation}
\mathcal{O}_{Y}\left(  V\right)  =A\left(  V^{\prime}\right)  ^{\mathbb{C}%
^{\ast}}%
\end{equation}
which is the subalgebra of elements invariant by this action. This is the
formal definition of a projective superspace (see for example
\cite{Varadarajan:note}), however we would like to construct the same space
more explicitly from gluing different superdomains as in sec. \ref{smfld}.

Let $X_{i}$ be the open set where the coordinate $z_{i}$ does not vanish. Then
the super-commutative ring $\mathcal{O}_{Y}\left(  X_{i}\right)  $ is
generated by elements of the type
\begin{equation}
f_{0}\left(  \frac{z_{0}}{z_{i}}, \dots, \frac{z_{i-1}}{z_{i}}, \frac{z_{i+1}%
}{z_{i}}, \dots,\frac{z_{n}}{z_{i}}\right)  \,,\quad
\end{equation}
\begin{equation}
f_{r}\left(  \frac{z_{0}}{z_{i}},...,\frac{z_{i-1}}{z_{i}},\frac{z_{i+1}%
}{z_{i}},...,\frac{z_{n}}{z_{i}}\right)  \frac{\theta^{r}}{z_{i}}\,,\quad
r=1,\dots,q\,.
\end{equation}

In fact, to be invariant with respect to the action of $\mathbb{C}^{\ast}$,
the functions $f_{I}$ in equation (\ref{nonzero}) must be homogeneous of
degree $-|I|$. Then, it is obvious that the only coordinate we can divide by,
on $X_{i}$, is $z_{i}$: all functions $f_{I}$ are of degree $-|I|$ and
holomorphic on $X_{i}$. If we put, on $X_{i}$, for $l\neq i$, $\Xi_{l}%
^{(i)}=\frac{z_{l}}{z_{i}}$ and $\Theta^{(i)}_{r}= \frac{\theta^{r}}{z_{i}}$,
then $\mathcal{O}_{Y}\left(  X_{i}\right)  $ is generated, as a
super-commutative ring, by the objects of the form
\begin{equation}
F_{0}^{(i)}\left(  \Xi_{0}^{(i)}, \Xi_{1}^{(i)},...,\Xi_{i-1}^{(i)}, \Xi
_{i+1}^{(i)},...,\Xi_{n}^{(i)}\right)  , \quad
\end{equation}
\begin{equation}
F_{a}^{(i)}\left(  \Xi_{0}^{(i)}, \Xi_{1}^{(i)},...,\Xi_{i-1}^{(i)}, \Xi
_{i+1}^{(i)},...,\Xi_{n}^{(i)}\right)  \Theta^{(i)}_{a},
\end{equation}
where $F_{0}^{(i)}$ and the $F_{a}^{(i)}$'s are analytic functions on
$\mathbb{C}^{n}$. In order to avoid confusion we have put the index $i$ in
parenthesis: it just denotes the fact that we are defining objects over the
local chart $X_{i}$. In the following, for convenience in the notation, we
also adopt the convention that $\Xi^{(i)}_{i} = 1$ for all $i$.




We have the two sheaves $\mathcal{O}_{Y}(X_{i})|_{X_{j}}$ and $\mathcal{O}%
_{Y}(X_{j})|_{X_{i}}$. In the same way as before, we have the morphisms given
by the \textquotedblleft coordinate changes\textquotedblright. So, on
$X_{i}\cap X_{j}$, the isomorphism simply affirms the equivalence between the
objects of the super-commutative ring expressed either by the first system of
affine coordinates, or by the second one. So for instance we have that
$\Xi_{l}^{(j)}=\frac{z_{l}}{z_{j}}$ and $\Theta_{r}^{(j)}=\frac{\theta^{r}%
}{z_{j}}$ can be also expressed as
\begin{equation}
\Xi_{l}^{(j)}=\frac{\Xi_{l}^{(i)}}{\Xi_{j}^{(i)}},\quad\Theta_{r}^{(j)}%
=\frac{\Theta_{r}^{(i)}}{\Xi_{j}^{(i)}}.
\end{equation}
Which, in the language used in the previous section, means that the morphism
$\psi_{ji}$ gluing $(X_{i}\cap X_{j},\mathcal{O}_{Y}(X_{i})|_{X_{j}})$ and
$(X_{j}\cap X_{i},\mathcal{O}_{Y}(X_{j})|_{X_{i}})$ is such that $\tilde{\psi
}_{ji}$ is the usual change of coordinates map on projective space and
\begin{equation}
\psi_{ji}^{\ast}(\Xi_{l}^{(j)})=\frac{\Xi_{l}^{(i)}}{\Xi_{j}^{(i)}},\quad
\psi_{ji}^{\ast}(\Theta_{r}^{(j)})=\frac{\Theta_{r}^{(i)}}{\Xi_{j}^{(i)}}%
\end{equation}

The super-manifold is obtained by observing that the coordinate changes
satisfy the cocycle conditions of the previous section.

\section{Integral forms and integration} \label{integral:form}

Most of supergeometry can be obtained straightforwardly by extending the
commuting geometry by means of the rule of signs, but this is not the case in
the theory of differential forms on supermanifolds. Indeed the naive notion of
\textquotedblleft superforms\textquotedblright\ obtainable just by adding a
$\mathbb{Z}_{2}$ grading to the exterior algebra turns out not to be suitable
for Berezin integration. In this section we briefly recall the definition of
"integral forms" and their main properties referring mainly to \cite{mare} for
a detailed exposition. The theory of superforms and their integration theory
has been widely studied in the literature and it is based on the notion of the
integral superforms (see for example \cite{manin} \cite{Voronov2}). The
 problem is that we can build the space $\Omega^k$ of $k$-superforms out of basic 1-superforms $d\theta^{i}$ and their wedge products, however these products are necessarily commutative, since the $\theta_i$'s are odd variables. Therefore, together
 with a differential operator $d$, the spaces $\Omega^k$ form a differential complex
\begin{equation}
0\overset{d}{\longrightarrow}\Omega^{0}\overset{d}{\longrightarrow}%
\Omega^{1}\dots\overset{d}{\longrightarrow}\Omega^{n}\overset
{d}{\longrightarrow}\dots\label{comA}%
\end{equation}
which is bounded from below, but not from above. In particular there is no notion 
of a top form to be integrated on the supermanifold $\mathbb{C}^{p+1|q}$.

The space of "integral forms" is obtained by adding to the usual space of superforms a new set of 
basic forms $\delta(d\theta)$, together with its "derivatives" $\delta^n(d\theta)$, and defining a product which satisfies certain formal properties. These properties are motivated and can be deduced from the following heuristic approach. In analogy with usual distributions acting on the space of smooth functions, we think of $\delta(d\theta)$ as an operator acting on the space of superforms as the usual Dirac's delta measure. We write this as
\[ \left\langle f(d\theta) ,\delta(d\theta) \right\rangle = f(0), \]
where $f$ is a superform. This means that $\delta(d\theta)$ kills all monomials in the superform $f$ which contain the term $d\theta$. The derivatives $\delta^{(n)}(d\theta)$ satisfy
\[ \left\langle f(d\theta) ,\delta^{(n)}(d\theta) \right\rangle = 
- \left\langle f'(d\theta), \delta^{(n-1)}(d\theta) \right\rangle = (-1)^n f^{(n)}(0), \]
like the derivatives of the usual Dirac $\delta$ measure. Moreover we can consider objects 
such as $g(d\theta) \delta(d\theta)$, which act by first multiplying by $g$ then applying $\delta(d\theta)$ (in analogy with a measure of type $g(x) \delta(x)$), and so on. The wedge products among these objects satisfy some simple relations such as (we will
always omit the symbol $\wedge$ of the wedge product):
\begin{eqnarray}
&& d x^I \wedge d x^J = - d x^J \wedge d x^I\,, \quad
d x^I \wedge d \theta^j =  d \theta^j \wedge d x^I\,, \quad \nonumber \\
&&d \theta^i \wedge d \theta^j =  d \theta^j \wedge d \theta^i\,, \quad
\delta(d\theta) \wedge \delta(d\theta^{\prime})=-\delta(d\theta^{\prime}%
) \wedge \delta(d\theta), \\ 
&&d\theta\delta(d\theta)=0\,, \quad d\theta\delta^{\prime}
(d\theta)=-\delta(d\theta).\label{comB} \nonumber %
\end{eqnarray}
The second and third property can be easily deduced from the above approach. 
To prove these formulas we observe the usual transformation property of the usual Dirac's delta 
function 
\begin{equation}\label{micA}
\delta (a x + b y ) \delta (c x + d y ) = 
\frac{1}{ Det \left(
\begin{array}{cc}
a  & b \\
c  & d   
\end{array}
\right)
} \delta(x) \delta(y)
\end{equation}
for $x,y \in \mathbb{R}$. By setting $a=0, b=1, c=1$ and $d=1$, the anticommutation property of 
Dirac's delta function of $d\theta$'s of (\ref{comB}) follows. 

We do not wish here to give an exhaustive and rigorous treatment of integral forms.  As we will see later, it is suffcient for our purposes that these simple rules give a well defined construction in the case of  superprojective spaces.  A systematic exposition of these rules can be found in \cite{integ} and they
can be put in a more mathematical framework using the results of
\cite{Bruzzo}. An interesting consequence of this procedure is the existence
of "negative degree" forms, which are those which reduce the degree of forms (e.g. $\delta'(d\theta)$ has degree $-1$). The integral forms could be also called
"pseudodifferential forms".
 
We introduce also the \textit{picture number} by counting the number of delta
functions (and their derivatives) and we denote by $\Omega^{r|s}$ the
$r$-forms with picture $s$. For example, in the case of $\mathbb{C}^{p+1|q}$,
the integral form
\begin{equation}
dx^{[K_{1}}\dots dx^{K_{l}]}d\theta^{( i_{l+1}}\dots d\theta^{i_{r})%
}\delta(d\theta^{[i_{r+1}})\dots\delta(d\theta^{i_{r+s}]})\ \label{comC}%
\end{equation}
is an $r$-from with picture $s$. All indices $K_{i}$  are
antisymmetrized among themselves, while the first $r- l$ indices are symmetric and the last 
$s+1$ are antisymmetrized. We denote by $[I_{1}\dots I_{s}]$ the
antysimmetrization of the indices and by $(i_1 \dots i_n)$ the symmetrization. 
Indeed, by also adding derivatives of
delta forms $\delta^{(n)}(d\theta)$, even negative form-degree can be considered, e.g. a form of
the type:
\begin{equation}
\delta^{(n_{1})}(d\theta^{i_{1}})\dots\delta^{(n_{s})}(d\theta^{i_{s}})
\label{comCA}%
\end{equation}
is a $-(n_{1}+\dots n_{s})$-form with picture $s$. Clearly $\Omega^{k | 0}$ is
just the set $\Omega^k$ of superforms, for $k \geq 0$.

We now briefly discuss how these forms behave under change of
coordinates, i.e. under sheaf morphisms. For very general type of morphisms it is necessary to
work with infinite formal sums in $\Omega^{r|s}$ as the following example
clearly shows.

Suppose $\Phi^{\ast}(\tilde{\theta}^{1})=\theta^{1}+\theta^{2}$ ,  $\Phi
^{\ast}(\tilde{\theta}^{2})=\theta^{2}$ be the odd part of a morphism. We want
to compute%
\begin{equation}
\Phi^{\ast}(\delta\left(  d\tilde{\theta}^{1}\right)  )=\delta\left(
d\theta^{1}+d\theta^{2}\right)
\end{equation}
in terms of the above forms. 
We can formally expand in series about, for example, $d\theta^{1}:$%
\begin{equation}
\delta\left(  d\theta^{1}+d\theta^{2}\right)  =\sum_{j}\frac{\left(
d\theta^{2}\right)  ^{j}}{j!}\delta^{(j)}(d\theta^{1})
\end{equation}
Recall that any usual superform is a polynomial in the $d\theta,$ therefore only a
finite number of terms really matter in the above sum, when we apply it to a superform.  Infact, applying the formulae above, we have for example,
\begin{equation}
\left\langle (d\theta^{1})^{k},\sum_{j}\frac{\left(  d\theta^{2}\right)  ^{j}%
}{j!}\delta^{(j)}(d\theta^{1})\right\rangle =(-1)^{k}(d\theta^{2})^{k}%
\end{equation}
Notice that this is equivalent to the effect of replacing $d\theta^{1}$ with
$-d\theta^{2}.$ We could have also interchanged the role of $\theta^{1}$ and
$\theta^{2}$ and  the result would be to replace $d\theta^{2}$ with
$-d\theta^{1}.$ Both procedures correspond precisely to the action we expect
when we apply the $\delta\left(  d\theta^{1}+d\theta^{2}\right)  $ Dirac
measure. We will not enter into more detailed treatment of other types of
morphisms, as this simple example will suffice. In the case of super-projective spaces the change of 
coordinate rule is simple and will be discussed in the next section. In the rest of the paper we 
will ignore the action $\langle \, , \, \rangle$ and do the computations following 
the above rules. 

We will see later, in Section \ref{deRham}, that integral forms form a new complex as follows
\begin{equation}
\dots\overset{d}{\longrightarrow}\Omega^{(r|q)}\overset{d}{\longrightarrow
}\Omega^{(r+1|q)}\dots\overset{d}{\longrightarrow}\Omega^{(p+1|q)}\overset
{d}{\longrightarrow}0 \label{comD}%
\end{equation}
where $\Omega^{(p+1|q)}$ is the top "form" $dx^{[K_{1}}\dots dx^{K_{p+1}%
]}\delta(d\theta^{\lbrack i_{1}})\dots\delta(d\theta^{i_{q}]})$ which can be
integrated on the supermanifold. As in the usual commuting geometry, there is
an isomorphism between the cohomologies $H^{(0|0)}$ and $H^{(p+1|q)}$ on a
supermanifold of dimension $(p+1|q)$. In addition, one can define two
operations acting on the cohomology groups $H^{(r|s)}$ which change the
picture number $s$ (see \cite{integ}).

Given a function $f(x,\theta)$ on the superspace $\mathbb{C}^{(p+1|q)}$, we
define its integral by the super top-form $\omega^{(p+1|q)}=f(x,\theta
)d^{p+1}x\delta(d\theta^{1})\dots\delta(d\theta^{q})$ belonging to
$\Omega^{(p+1|q)}$ as follows
\begin{equation}
\int_{\mathbb{C}^{(p+1|q)}}\omega^{(p+1|q)}=\epsilon^{i_{1}\dots i_{q}%
}\partial_{\theta^{i_{1}}}\dots\partial_{\theta^{i_{q}}}\int_{\mathbb{C}%
^{p+1}}f(x,\theta) \label{comE}%
\end{equation}
where the last equalities is obtained by integrating on the delta functions
and selecting the bosonic top form. The remaining integrals are the usual
integral of densities and the Berezin integral. The latter can be understood
in terms of the Berezinian sheaf \cite{bere-sheaf}. It is easy to show that
indeed the measure is invariant under general coordinate changes and the
density transform as a Berezinian with the superdeterminant.

\section{\v{C}ech cohomology of $\mathbb{P}^{1|1}$}

We describe now \v{C}ech cohomology on super-projective spaces, with respect
to this particular sheaf of "integral $1$-forms".

We will begin by considering $\mathbb{P}^{1|1}$. 
$\mathbb{P}^{1}$ has a natural covering with two charts, $U_{0}$ and $U_{1}$,
where
\begin{equation}
U_{0}=\{[z_{0};z_{1}]\in\mathbb{P}^{1}:z_{0}\neq0\},
\end{equation}
\begin{equation}
U_{1}=\{[z_{0};z_{1}]\in\mathbb{P}^{1}:z_{1}\neq0\}.
\end{equation}
The affine coordinates are $\gamma=\frac{z_{1}}{z_{0}}$ on $U_{0}$ and
$\widetilde{\gamma}=\frac{z_{0}}{z_{1}}$ on $U_{1}$. The odd generators are
$\psi$ on $U_{0}$ and $\widetilde{\psi}$ on $U_{1}$. The gluing morphism of
sheaves on the intersection $U_{0}\cap U_{1}$ has pull-back given by:
\begin{equation}
\Phi^{\ast}:\mathcal{O}(U_{0}\cap U_{1})[\psi]\longmapsto\mathcal{O}(U_{0}\cap
U_{1})[\widetilde{\psi}]
\end{equation}
with the requirement that:
\begin{equation}
\Phi^{\ast}(\gamma)=\frac{1}{\widetilde{\gamma}},\Phi^{\ast}(\psi
)=\frac{\widetilde{\psi}}{\widetilde{\gamma}}.
\end{equation}

We now consider a sheaf of differential on $\mathbb{P}^{1|1}$. As we
already said in the previous section, we must add objects of the
type "$d\gamma$" and of the type "$d\psi$" on $U_{0}$. But $d\psi$
is an even generator, because $\psi$ is odd, so we are not able to
find a differential form of maximal degree. We introduce then
the generator $\delta(d\psi)$, which allow us to perform integration in the
"variable" $d\psi$. It satisfies the rule $d\psi \delta(d\psi)=0$.
This means that $\delta(d\psi)$ is like a Dirac measure on the space of the
analytic functions in $d\psi$ which gives back the evaluation at
zero. We must also introduce the derivatives $\delta^{(n)}(d\psi)$,
where $d\psi \delta^{\prime}(d\psi)=-\delta(d\psi)$, and, in
general, $d\psi\delta ^{(n)}(d\psi)=-\delta^{(n-1)}(d\psi)$. In this
way, the derivatives of the delta represent anticommuting
differential forms of negative degree.

Let's define the following sheaves of modules:
\begin{equation}
\Omega^{0|0}(U_{0})=\mathcal{O}(U_{0})[\psi];
\end{equation}
\begin{equation}
\Omega^{1|0}(U_{0})=\mathcal{O}(U_{0})[\psi]d\gamma\oplus\mathcal{O}%
(U_{0})[\psi]d\psi;
\end{equation}
and similarly un $U_{1}$. The general sheaf $\Omega^{n|0}$ is
locally made up
by objects of the form%
\begin{equation}
\mathcal{O}(U_{0})[\psi](d\gamma)^{i}(d\psi)^{j},
\end{equation}
where $i=0;1$ and $i+j=n$. The definitions on $U_{1}$ are similar,
the only difference is that we will use the corresponding
coordinates on $U_{1}$. Note that $\Omega^{n|0}$ is non zero for all integers $n \geq 0$.

We also define the sheaves of modules $\Omega^{l|1}$, which, on $U_{0}$,
contain elements of the form:%
\begin{equation}
\mathcal{O}(U_{0})[\psi](d\gamma)^{i}\delta^{(j)}(d\psi),
\end{equation}
with $i-j=l$. The elements containing $"d\psi"$ cannot appear, since
they cancel with the delta forms. On $U_{1}$, the sections of this
sheaf assume a similar structure with respect to the coordinates on
$U_{1}$.

Notice that $\Omega^{l|1}$ is non zero for all integers $l$ with $l \leq 1$, in particular for all negative integers. 
We still have to describe coordinate change in the intersection
$U_{0}\cap U_{1}$ of the objects $\{d\gamma,d\psi,\delta(d\psi)\}$.
They are given by:
\begin{equation}
\Phi^{\ast}d\widetilde{\gamma}=-\frac{1}{\gamma^{2}}d\gamma,
\end{equation}
and
\begin{equation} \label{pullpsi}
\Phi^{\ast}d\widetilde{\psi}=\frac{d\psi}{\gamma}-\frac{d\gamma \, \psi}%
{\gamma^{2}}.%
\end{equation}
More generally, for any integer $n >0$, we have the formula
\begin{equation} \label{pullpsi_n}
\Phi^{\ast}(d\widetilde{\psi})^n= \left( \frac{d\psi}{\gamma} \right)^n -\frac{d\gamma \, \psi}%
{\gamma^{2}} \left( \frac{d\psi}{\gamma} \right)^{n-1}.%
\end{equation}
It only remains to compute how $\delta(d\psi)$ transforms in a
coordinate change. We can proceed as outlined in the previous
section.

In this case, we write:
\begin{equation}
\Phi^{\ast}\delta(d\widetilde{\psi})=\delta\left(  \frac{d\psi}{\gamma}%
-\frac{d\gamma\,\psi}{\gamma^{2}}\right)
\end{equation}
Then:%
\begin{equation}
\Phi^{\ast}\delta(d\widetilde{\psi})=\gamma\delta\left(  d\psi-\frac
{d\gamma\psi}{\gamma}\right)  =\gamma\delta\left(  d\psi\right)  -\gamma
\frac{d\gamma\,\psi}{\gamma}\delta(d\psi)=\gamma\delta\left(  d\psi\right)
- \psi d\gamma \delta^{\prime}(d\psi).
\end{equation}
Notice that the latter equation, together with (\ref{pullpsi}), implies that $$\Phi^*(d \tilde{\psi} \delta(d\tilde{\psi})) = 0$$ as expected. 

Hence the generator $\delta(d\tilde{\psi})$ and its properties are well defined. Similarily, one can compute that the derivatives $\delta^n(d\tilde{\psi})$ satisfy the following change of coordinates formula
\begin{equation} \label{pulldeltan}
 \Phi^{\ast}\delta^n(d\widetilde{\psi}) = \gamma^{n+1} \delta^n\left(  d\psi\right)
- \gamma^{n} \psi \, d\gamma \, \delta^{n+1}(d\psi).
\end{equation}
Now, we can proceed in calculating sheaf cohomology for each of the
sheaves $\Omega^{i|j}$ with respect to the covering $\{ U_{0};U_{1}
\}$.
\begin{theorem}
The covering $\{U_{0};U_{1}\}$ is acyclic with respect to each of the sheaves
$\Omega^{i|j}$.
\end{theorem}
\textbf{Proof.} We know that $U_{0}$ and $U_{1}$ are both isomorphic
to $\mathbb{C}$, while $U_{0} \cap U_{1}$ is isomorphic to
$\mathbb{C}^{*}$. Moreover, we know that, classically,
$H^{q}(\mathbb{C}; \mathcal{O})=\{ 0 \}$, and that
$H^{q}(\mathbb{C}^{*}; \mathcal{O})=0$. We note that the restriction
to each open set of the sheaf $\Omega^{i|j}$ is simply the direct
sum of the sheaf $\mathcal{O}$ a certain finite number of times.

For example,
\begin{equation}
\Omega^{1|1}{(U_{0}\cap U_{1})}=\mathcal{O}(\mathbb{C}^{\ast})d\gamma
\delta(d\psi)+\mathcal{O}(\mathbb{C}^{\ast})\psi d\gamma\delta(d\psi).
\end{equation}
Note that the symbols $d\gamma\delta(d\psi)$ and $\psi d\gamma\delta(d\psi)$
represent the generators of a vector space, then, each of the direct summands
can be treated separately. So, we see that a chain of $\Omega^{i|j}$ (on
$\mathbb{C}$ or $\mathbb{C}^{\ast}$) is a cocycle if and only if each of the
summands is a cocycle, and it is a coboundary if and only if every summand is
a coboundary.

We now begin the computation of the main cohomology groups on $\mathbb{P}%
^{1|1}$. For $\check{H}^{0}$ we have the following result:

\begin{theorem} \label{h0} For integers $n \geq 0$, the following isomorphisms hold
\[ \check{H}^{0}(\mathbb{P}^{1|1}, \Omega^{n|0}) \cong \begin{cases}
                                                                                       0, \quad n> 0, \\
                                                                                       \mathbb{C}, \quad n=0.
                                                                                  \end{cases} \]   
 \[ \check{H}^{0}(\mathbb{P}^{1|1}, \Omega^{- n|1}) \cong \mathbb{C}^{4n+4},  \]
 \[ \check{H}^{0}(\mathbb{P}^{1|1}, \Omega^{1|1}) \cong 0 \]
\end{theorem}

\begin{proof}

\begin{itemize}
\item Let's begin from $\check{H}^{0}(\mathbb{P}^{1|1}, \Omega^{0|0})$. On
$U_{1}$, the sections of the sheaf have the structure:%
\begin{equation}
f(\widetilde{\gamma})+f_{1}(\widetilde{\gamma})\widetilde{\psi}.
\end{equation}

On the intersection $U_{0}\cap U_{1}$ they transform in the following way:%
\begin{equation}
f\left(  \frac{1}{\gamma}\right)  +\frac{\psi}{\gamma}f_{1}\left(  \frac
{1}{\gamma}\right)
\end{equation}

So, the only globally defined sections (i.e which can be extended
also on
$\mathbb{P}^{1|1}$) are the constants:%
\begin{equation}
\check{H}^{0}(\mathbb{P}^{1|1},\Omega^{0|0}) \cong \mathbb{C}.
\end{equation}

\item Let's consider $\check{H}^{0}(\mathbb{P}^{1|1}, \Omega^{n|0})$, with $n > 0$. On
$U_{1}$, the sections of the sheaf have the structure:%
\begin{equation} \label{forms_n0}
\left( f_{0}(\widetilde{\gamma}) +f_{1}(\widetilde{\gamma
})\widetilde{\psi} \right) \, d\widetilde{\gamma} \, (d \widetilde{\psi})^{n-1} + \left( f_{2}(\widetilde{\gamma}) +f_{3}(\widetilde{\gamma})\widetilde{\psi} \right) \,  (d\widetilde{\psi})^n.%
\end{equation}

Since both $d\widetilde{\gamma}$ and $d\widetilde{\psi}$ transform, by
coordinate change, producing a term ${1/}{\gamma^{2}}$, none of these sections
can be extended on the whole $\mathbb{P}^{1|1}$, except the zero section. So,
\begin{equation}
\check{H}^{0}(\mathbb{P}^{1|1}, \Omega^{n|0}) \cong 0.
\end{equation}

\item Let us now compute $\check{H}^{0}(\mathbb{P}^{1|1};\Omega^{-n|1})$ for every 
integer $n \geq 0$. On $U_1$, the sections of the sheaf have the form:
 \begin{equation} \label{forms_n1}
\left( f_{0}(\widetilde{\gamma}) +f_{1}(\widetilde{\gamma
})\widetilde{\psi} \right) \, \delta^n(d\widetilde{\psi})+ \left( f_{2}(\widetilde{\gamma
}) +f_{3}(\widetilde {\gamma}) \widetilde{\psi} \right) \, d\widetilde{\gamma} \, \delta^{n+1}(d\widetilde{\psi
}). 
\end{equation}
Using the change of coordinates formula (\ref{pulldeltan}) one can verify that 
on the intersection $U_{0}\cap U_{1}$ they transform in the following way:%
\begin{align}
& \left( f_{0}\left(  \frac{1}{\gamma}\right)   + f_{1}\left(  \frac{1}{\gamma
}\right)  \frac{\psi}{\gamma} \right) \left(  \gamma^{n+1} \delta^{n}\left(  d\psi\right)
- \gamma^{n} \psi \, d\gamma \, \delta^{n+1}(d\psi) \right)  - \nonumber \\
- & \left( f_{2}\left(  \frac{1}{\gamma}\right)  +f_{3}\left(  \frac{1}{\gamma}\right)  \frac{\psi}{\gamma} \right) \frac{d \gamma}%
{\gamma^{2}} \left(  \gamma^{n+2} \delta^{n+1}\left(  d\psi\right)
- \gamma^{n+1} \psi \, d\gamma \, \delta^{n+2}(d\psi) \right)  = \nonumber \\
= & \left( f_{0}\left(  \frac{1}{\gamma}\right) \gamma^{n+1}  + f_{1}\left(  \frac{1}{\gamma
}\right)  \gamma^{n} \psi  \right)    \delta^{n}\left(  d\psi \right) - \nonumber \\
- &  \left( f_{2} \left( \frac{1}{\gamma}\right) \gamma^n  + \left(  f_{0} \left(  \frac{1}{\gamma} \right) \gamma^n  + f_{3}\left(  \frac{1}{\gamma}\right) \gamma^{n-1} \right)  \psi  \right) d \gamma \, \delta^{n+1}\left(  d\psi \right)
\end{align}
Therefore this expression extends to a global section 
if and only if  the following conditions hold. The coefficient $f_0$ is a polynomial of degree $n+1$, while $f_1$, $f_2$ and $f_3$ are polynomials of degree $n$. Moreover, if $a_{n+1}$ and $b_{n}$ are the coefficients of maximal degree in $f_0$ and $f_3$ respectively, then $a_{n+1} = -b_{n}$. This establishes that $\check{H}^{0}(\mathbb{P}^{1|1}, \Omega^{-n|1})$ has dimension $4n+4$.

\item Let's consider $\check{H}^{0}(\mathbb{P}^{1|1};\Omega^{1|1})$. On
$U_{1}$, the sections of the sheaf have the structure:%
\begin{equation}
\left( f_{0}(\widetilde{\gamma}) +f_{1}(\widetilde{\gamma})\widetilde{\psi} \right) d\widetilde{\gamma}\delta
(d\widetilde{\psi})
\end{equation}

These sections cannot be defined on the whole $\mathbb{P}^{1}$, since they
transform as:%
\begin{eqnarray*}
 - \left( f_{0}\left(  \frac{1}{\gamma}\right) + f_{1}\left(  \frac{1}{\gamma}\right)  \frac{\psi}{\gamma} \right) %
\frac{d\gamma}{\gamma^{2}} \left( \gamma\delta\left(  d\psi\right)
- \psi d\gamma \delta^{\prime}(d\psi) \right) & = & \  \\ 
= - \left( f_{0}\left(  \frac{1}{\gamma}\right) \frac{1}{\gamma} + f_{1}\left(  \frac{1}{\gamma}\right)  \frac{\psi}{\gamma^2} \right) d\gamma \, \delta\left(  d\psi\right). & \ & \ 
\end{eqnarray*}
So,%
\begin{equation}
\check{H}^{0}(\mathbb{P}^{1|1},\Omega^{1|1})= 0 .
\end{equation}

\end{itemize}
\end{proof}

A similar computation can be done to obtain the groups $\check{H}%
^{1}(\mathbb{P}^{1|1};\Omega^{i|j})$. The elements of the \v{C}ech cohomology
are sections $\sigma_{01}$ of $\Omega_{|U_{0}\cap U_{1}}^{i|j}$ which cannot
be written as differences $\sigma_{0}-\sigma_{1}$, with $\sigma_{0}$ defined
on $U_{0}$ and $\sigma_{1}$ defined on $U_{1}$. We have the following result:

\begin{theorem} For integers $n \geq 0$, the following isomorphisms hold
\[ \check{H}^{1}(\mathbb{P}^{1|1}, \Omega^{n |0}) \cong \mathbb{C}^{4n}
                                                                      \]   
 \[ \check{H}^{1}(\mathbb{P}^{1|1}, \Omega^{- n|1}) \cong 0,  \]
 \[ \check{H}^{1}(\mathbb{P}^{1|1}, \Omega^{1|1}) \cong \mathbb{C} \]
\end{theorem}

\begin{proof}

\begin{itemize}
\item $\check{H}^{1}(\mathbb{P}^{1|1}, \Omega^{0|0})=\{0\}$, since for every
section on $U_{1}\cap U_{0}$ we have the structure:
\begin{equation}
f(\widetilde{\gamma})+f_{1}(\widetilde{\gamma})\widetilde{\psi},
\end{equation}
we can decompose the Laurent series of $f$ and $f_{1}$ in a singular
part and in a holomorphic component. The singular part is defined on
$U_{0}$, while the holomorphic part is defined on $U_{1}$. So, it's
easy to write every section of $\Omega^{0|0}$ on $U_{0}\cap U_{1}$
as a difference of sections on $U_{0}$ and $U_{1}$.

\item We now compute $\check{H}^{1}(\mathbb{P}^{1|1}, \Omega^{n|0})$ for $n>0$.
A section on $U_0$ is of the type
\[ \left( f_{0}(\gamma) +f_{1}(\gamma) \psi \right) \, d \gamma \, (d \psi)^{n-1} + \left( f_{2}(\gamma) +f_{3}(\gamma) \psi \right) \,  (d\psi)^n.
 \]
 while a section on $U_1$ is of the type
\[ \left( g_{0}(\widetilde{\gamma}) +g_{1}(\widetilde{\gamma
})\widetilde{\psi} \right) \, d\widetilde{\gamma} \, (d \widetilde{\psi})^{n-1} + \left( g_{2}(\widetilde{\gamma}) +g_{3}(\widetilde{\gamma})\widetilde{\psi} \right) \,  (d\widetilde{\psi})^n.%
 \]
All functions here are regular. A computation shows that, taking the difference of the two on $U_0 \cap U_1$ and expressing everything in the coordinates $\gamma$ and $\psi$, gives us an expression of the type
\begin{align}
 & \left( f_0(\gamma) + g_0(\gamma^{-1}) \gamma^{-(n+1)} \right) \, d \gamma \, (d \psi)^{n-1} +  \nonumber \\
 + & \left( f_1(\gamma) + g_1(\gamma^{-1}) \gamma^{-(n+2)} + g_2(\gamma^{-1}) \gamma^{-(n+1)} \right) \, \psi d \gamma \, (d \psi)^{n-1} +  \nonumber \\
 + & \left( f_2(\gamma) - g_2(\gamma^{-1}) \gamma^{-n} \right) \, (d \psi)^{n} + \nonumber \\
 + & \left( f_3(\gamma) - g_3(\gamma^{-1}) \gamma^{-(n+1)} \right) \, \psi (d \psi)^{n}. \nonumber 
\end{align}
It is clear that in the first row there are no terms of the type $a_k \gamma^{-k}$ with $1 \leq k \leq n$, so this gives us $n$ parameters for an element of $\check{H}^{1}(\mathbb{P}^{1|1}, \Omega^{n|0})$. Similarily, the second row gives us $n$ parameters, the third gives us $n-1$ and the fourth $n$. This gives a total of $4n-1$. Notice now that in the above expression the coefficient of $\gamma^{-(n+1)}$ in the second row must be equal to the coefficient of $\gamma^{-n}$ in the third row. This constraint on the terms of the above type gives us room for an extra parameter in the elements of $\check{H}^{1}(\mathbb{P}^{1|1}, \Omega^{n|0})$. We therefore have a total of $4n$ parameters. 

\item We compute in a similar way $\check{H}^{1}(\mathbb{P}^{1|1} , \Omega^{-n|1})$ for $n \geq 0$. A computation shows that a difference between a section on $U_0$ and a section on $U_1$ is of the type 
\begin{align}
 & \left( f_0(\gamma) - g_0(\gamma^{-1}) \gamma^{n+1} \right) \, \delta^{n}(d \psi) +  \nonumber \\
 + & \left( f_1(\gamma) - g_1(\gamma^{-1}) \gamma^{n} \right) \, \psi \delta^{n}(d \psi) +  \nonumber \\
 + & \left( f_2(\gamma) + g_2(\gamma^{-1}) \gamma^{n} \right) \, d \gamma \, \delta^{n+1}(d \psi)  + \nonumber \\
 + & \left( f_3(\gamma) + g_0(\gamma^{-1}) \gamma^n +  g_3(\gamma^{-1}) \gamma^{n-1} \right) \, \psi d\gamma \, \delta^{n+1}(d \psi). \nonumber 
\end{align}
It is clear that every section on $U_0 \cap U_1$ is represented in such an expression. Therefore we have $\check{H}^{1}(\mathbb{P}^{1|1} , \Omega^{-n|1}) = 0$

\item We see in a similar way that 
$\check{H}^{1}(\mathbb{P}^{1|1};\Omega^{1|1})=\mathbb{C}$, in fact the section on $U_0 \cap U_1$ which are 
not differences are all generated by 
\begin{equation}  \label{generatore}
\frac{\psi
d\gamma\delta(d\psi)}{\gamma}. 
\end{equation}
\end{itemize}
This completes the proof. \end{proof}

Notice that $\check{H}^{1}(\mathbb{P}^{1|1}, \Omega^{n+1 |0})$ and $\check{H}^{0}(\mathbb{P}^{1|1}, \Omega^{- n|1})$ have the same dimension. There is an interesting explanation of this fact, in fact we can construct a pairing 
\[ \check{H}^{1}(\mathbb{P}^{1|1}, \Omega^{n+1 |0}) \times \check{H}^{0}(\mathbb{P}^{1|1}, \Omega^{- n|1}) \rightarrow  
           \check{H}^{1}(\mathbb{P}^{1|1}, \Omega^{1|1}) \cong \mathbb{C}\]
as follows. As explained above, an element of $\check{H}^{1}(\mathbb{P}^{1|1}, \Omega^{n+1 |0})$ is of the type 
\begin{equation} \label{term1}
\left( f_{0}(\gamma^{-1}) +f_{1}(\gamma^{-1}) \psi \right) \, d \gamma \, (d \psi)^{n} + \left( f_{2}(\gamma^{-1}) +f_{3}(\gamma^{-1}) \psi \right) \,  (d\psi)^{n+1}.
 \end{equation}
where $f_0$ and $f_1$ are polynomials of degree at most $n+1$, while $f_1$ and $f_2$ can be chosen to be respectively of degree at most $n+2$ and $n$ or both of degree at most $n+1$.  An element of $\check{H}^{0}(\mathbb{P}^{1|1}, \Omega^{- n|1})$ is of the type 

\begin{equation} \label{term2}
\left( g_{0}(\gamma) +g_{1}(\gamma) \psi \right) \, \delta^n(d\psi)+ \left( g_{2}(\gamma) +g_{3}(\gamma) \psi \right) \, d\gamma \, \delta^{n+1}(d\psi), 
\end{equation}

where $g_{0}$ is a polynomial of degree $n+1$, $g_1 \ldots, g_3$ are polynomials of degree $n$ and the coefficients of maximal degree in $g_0$ and $g_3$ are opposite to each other.  Now recall that we have a pairing
\[ \Omega^{n+1 | 0} \times \Omega^{ - n | 1} \rightarrow  \Omega^{ 1 | 1}\]
obeying the rules explained in Section \ref{integral:form}. For instance 
\[ \langle d \gamma \, (d \psi)^{n} ,  \delta^n(d\psi) \rangle = (-1)^n n! \, d \gamma \, \delta(d\psi), \]
\[ \langle  (d \psi)^{n+1} ,  d \gamma \,\delta^{n+1}(d\psi) \rangle  = - (-1)^n (n+1)! \, d \gamma \, \delta(d\psi), \]
\[ \langle d \gamma \, (d \psi)^{n} , d \gamma \,\delta^{n+1}(d\psi) \rangle = \langle  (d \psi)^{n+1} , \delta^n(d\psi) \rangle = 0.\]
It can be checked that this product descends to a pairing in cohomology. We have the following 
\begin{lemma}
On $\mathbb{P}^{1|1}$ the above product in cohomology is non-degenerate. 
\end{lemma}
\begin{proof} The product between (\ref{term1}) and (\ref{term2}) is cohomologous to the expression
\begin{equation} \label{molt}
(-1)^n n! \left( (f_0g_1 + f_1g_0) - (n+1)( f_2g_3 + f_3 g_2) \right) \psi \, d \gamma \, \delta(d \psi). 
\end{equation}
We have to prove that if (\ref{term2}) is arbitrary and non zero, then we can chose $f_0, \ldots, f_3$ so that the above expression is cohomologous to (\ref{generatore}). We can assume one of the $g_0, \ldots, g_3$ to be non zero. If $g_0 \neq 0$, let $a_k$ be the coefficient of highest degree in $g_0$, hence $k \leq n+1$.  Define 
\[ f_1 = C \gamma^{-{k+1}},\] 
and $f_0, f_2, f_3$ to be zero. Then, for suitably chosen $C \neq 0$ we can easily see that (\ref{molt}) is cohomologous to (\ref{generatore}). Notice also that $k+1 \leq n+2$, so the choice of $f_0, \ldots, f_3$ gives a well defined element of $\check{H}^{1}(\mathbb{P}^{1|1}, \Omega^{n+1 |0})$. Similar arguments hold when $g_1, g_2$ or $g_3$ are not zero.
\end{proof}

A consequence of this lemma is that  $\check{H}^{1}(\mathbb{P}^{1|1}, \Omega^{n+1 |0})$ and $\check{H}^{0}(\mathbb{P}^{1|1}, \Omega^{- n|1})$ are dual to each other. This explains why they have the same dimension. 

\section{Super de Rham Cohomology.} \label{deRham}

We now briefly describe smooth and holomorphic de Rham cohomology with respect to the $d$ differential 
on superforms.

On a fixed complex supermanifold $M^{n|m}$ we denote by $\mathcal{A}^{i|j}$ and $\Omega^{i|j}$ 
respectively the sheaf of smooth and holomorphic superforms of degree $i$  with picture number $j$ and by $\mathbf{A}^{i|j}$  and $\mathbf{\Omega}^{i|j}$ the global sections of these sheaves. As usual for superforms, $i$ can also have negative values. 
On $\mathbf{A}^{*|j}$ (or locally on $\mathcal{A}^{*|j}$) we can define the exterior differential operator $d:
\mathbf{A}^{i|j} \rightarrow\mathbf{A}^{i+1|j}$ which satisfies the following rules:

\begin{itemize}
\item[1.)] $d$ behaves as a differential on functions;

\item[2.)] $d^{2}=0$;

\item[3.)] $d$ commutes with $\delta$ and its derivatives, and so $d (\delta^{(k)}(d \psi))=0$.
\end{itemize}

Similarily, the same operator $d$ is defined on 
$\mathbf{\Omega}^{*|j}$, and behaves as the $\partial$ operator on holomorphic functions (since $\overline \partial$ always vanishes). 

It is easy to verify that, on the intersection of $2$ charts, $d$ commutes
with the pull-back map $\Phi^{\ast}$ expressing the "coordinate changes". This
is due to the particular definition of the pull-back of the differentials, and
it implies that $d$ is well defined and it does not depend on coordinate systems. 

As an example, we prove it on $\mathbb{P}^{1|1}$ in the holomorphic case, leaving to the reader
the easy generalization to every other super-projective space.

\begin{itemize}
\item We know that $\Phi^{\ast}(\widetilde{\gamma})=\frac{1}{\gamma}$, so it's
easy to see that $d\left(  \frac{1}{\gamma}\right)  =\Phi^{\ast}%
d(\widetilde{\gamma})=-\frac{1}{\gamma^{2}}d\gamma$.

\item We know that $\Phi^{\ast}(\widetilde{\psi})=\frac{\psi}{\gamma}$, so
it's easy to see that $d\left(  \frac{\psi}{\gamma}\right)  =\Phi^{\ast
}d(\widetilde{\psi})=-\frac{1}{\gamma^{2}}d\gamma\,\psi+\frac{d\psi}{\gamma}$.

\item We know that $\Phi^{\ast}\delta(d\widetilde{\psi})=\gamma\delta\left(
d\psi\right)  -d\gamma\,\psi\delta^{\prime}(d\psi)$. Then, $\Phi^{\ast
}d(\delta(d\widetilde{\psi}))=0$.

But, $d(\Phi^{\ast}\delta(d\widetilde{\psi}))=d(\gamma\delta\left(
d\psi\right)  -d\gamma\,\psi\delta^{\prime}(d\psi))=d\gamma\,\delta\left(
d\psi\right)  +d\gamma\,d\psi\,\delta^{\prime}(d\psi))=0$.
\end{itemize}

Now $( \mathbf{A}^{*|j}(M), d)$ and $( \mathbf{\Omega}^{*|j}(M), d)$ define complexes, whose cohomology groups we call
respectively the smooth and holomorphic super de Rham cohomology groups:

\begin{definition}
If $Z^{i|j}$ is the set of the $d$-closed forms in $\mathbf{A}^{i|j}$, and
$B^{i|j}= d \mathbf{A}^{i-1|j}$. Then, the $i|j$-th smooth de Rham cohomology group is
the quotient of additive groups:
\begin{equation}
H_{DR}^{i|j}(M^{n|m})=\frac{Z^{i|j}}{B^{i|j}}. 
\end{equation}
Similarily we define the holomorphic de Rham cohomology groups which we denote by $H_{DR}^{i|j}(M^{n|m}, hol)$
\end{definition}

We now calculate the holomorphic super de Rham cohomology of $\mathbb{C}^{m|n}$.

Let's call $\{\gamma_{1}, \gamma_{2},...,\gamma_{m}\}$ the even coordinates
and $\{ \psi_{1}, \psi_{2}, ..., \psi_{n}\}$ the odd coordinates of
$\mathbb{C}^{m|n}$.

Clearly the following forms are closed:

\begin{itemize}
\item[a)] $1$;

\item[b)] $\{d\gamma_{i}\}$, $i\in\{1;2;...,m\}$;

\item[c)] $\{d\psi_{j}\}$, $j\in\{1;2;...;n\}$;

\item[d)] $\{d\gamma_{h}\cdot\psi_{k}+\gamma_{h}d\psi_{k}=d(\gamma_{h}%
\cdot\psi_{k})\}$, $h\in\{1;2;...,m\}$, $k\in\{1;2;...,n\}$;

\item[e)] $\{\delta^{(k)}(d\psi_{a})\}$, $a\in\{1;2;...;n\}$ and $k \in \mathbb{N}$;

\item[f)] $\{\psi_{b}\delta(d\psi_{b})\}$, $b\in\{1;2;...,n\}$.
\end{itemize}
All other closed forms are products and linear combinations of these with coefficients some holomorphic 
functions in the even coordinates.  Observe that  $\{\psi_{b}\delta
(d\psi_{b})\}$, with $b\in\{1;2;...,n\}$ are not exact. A calculation shows that  the holomorphic super de Rham cohomology
$H^{i|j}(\mathbb{C}^{m|n}, hol)$ is zero whenever $i > 0$, it is generated by $1$
when $i =j=0$, by $\{\psi_{b}\delta(d\psi_{b})\}$ when $i=0$ and $j =1$ and by
their $j$-th exterior products when $i=0$ and $j \geq2$. Similarily we can compute the 
smooth de Rham cohomology of  $\mathbb{R}^{m|n}$.

\begin{remark}
In particular, we see that the super-vector space $\mathbb{C}^{m|n}$ (or $\mathbb{R}^{m|n}$) 
does not satisfy the Poincar\`{e} lemma, since its de Rham cohomology is not trivial.
The forms $\{\psi_{i}\delta(d\psi_{i})\}$ can be seen as even generators of
the "odd component" of the cohomology.
\end{remark}

As an example we compute the holomorphic de Rham cohomology of $\mathbb{P}^{1|1}$. We have:

\begin{theorem} For $n \geq 0$, the holomorphic de Rham cohomology groups of $\mathbb{P}^{1|1}$ are as follows:

\[
H_{DR}^{n|0}(\mathbb{P}^{1|1}, hol) \cong \begin{cases}
                                                                                       0, \quad n> 0, \\
                                                                                       \mathbb{C}, \quad n=0.
                                                                                  \end{cases} \] 
\[
H_{DR}^{- n|1}(\mathbb{P}^{1|1}, hol) \cong \begin{cases}
                                                                                       0, \quad n> 0, \\
                                                                                       \mathbb{C}, \quad n=0.
                                                                                  \end{cases} \]
\[
H_{DR}^{1|1}(\mathbb{P}^{1|1}, hol) \cong 0. \]
\end{theorem}

\begin{proof} We have given explicit descriptions of global sections of the sheaves $\Omega^{i|j}$ in Theorem ÷\ref{h0} and 
therefore it is a rather straightforward computation to determine which forms are closed and which are exact in terms of the coefficients describing the forms (see formulas (\ref{forms_n0}) and (\ref{forms_n1})). We leave the details to the reader. Notice that 
$H_{DR}^{0|1}(\mathbb{P}^{1|1}, hol)$ is generated by the closed form $\psi \delta(d\psi)$ which is globally defined on $\mathbb{P}^{1|1}$.
\end{proof}

Now consider a general smooth super manifold $M^{n|m}$. On $M$ we can define the pre-sheaf which associates to every open subset $U \subset M$ the smooth super de Rham $i|j$-cohomology group of $U^{n|m}$ and we denote
the corresponding sheaf by $\mathcal{H}^{i|j}$. If follows from the above
remark that $\mathcal{H}^{i|j}$ is the constant $\mathbb{C}$-sheaf when
$i,j=0$, a non zero sheaf when $i=0$ and $j > 0$ and the zero sheaf
otherwise. It makes therefore sense to consider the \v{C}ech cohomology groups
which we denote by $\check{H}^{p}(M^{n|m}, \mathcal{H}^{i|j})$ (which are
zero when $i > 0$). Recall that a \textbf{good cover} is an open covering
$U_{\alpha}$ of $M$ such that every non-empty finite intersection
$U_{\alpha_{0}}\cap U_{\alpha_{1}}\cap...\cap U_{\alpha_{p}}$ is diffeomorphic
to $\mathbb{R}^{n}$. We can now prove a generalization of the classical
equivalence of \v{C}ech and De Rham cohomology

\begin{theorem}
Given a supermanifold $M^{n|m}$, for $i \geq 0$ we have the following isomorphism
\begin{equation}
H_{DR}^{i|j}(M^{n|m}) \cong \check{H}^{i}(M^{n|m}, \mathcal{H}^{0|j})
\end{equation}

\end{theorem}

\begin{proof}
For the proof we can use the same method used in \cite{BottTu} for the
classical equivalence of \v{C}ech and De Rham cohomology. Let us fix a good
cover $\underline{\mathcal{U}} = \{ U_{\alpha} \}$ of $M$. For integers $p,q
\geq0$, let us set
\begin{equation}
K^{p,q} = \mathcal{C}^{p}( \mathcal{A}^{q|j}, \underline{\mathcal{U}} ),
\end{equation}
where the righthand side denotes the usual $p$-cochains of the sheaf
$\mathcal{A}^{q|j}$, with respect to the covering $\underline{\mathcal{U}}$.
Then we can form the double complex $(K, d, \delta)$, where $K = \oplus_{p,q
\geq0} K^{p,q}$ and the operators are the usual exterior differential operator
$d$ and the \v{C}ech co-boundary operator $\delta$. From this double complex
one can construct two spectral sequences $(E^{p,q}_{r}, d_{r})$ and
$(E^{^{\prime}p,q}_{r}, d_{r})$ both converging to the total cohomology
$H_{D}(K)$ of the double complex (see \cite{BottTu}). We have that
\begin{equation}
E^{p,q}_{2} = \check{H}^{p}(H_{DR}^{q|j}( \mathcal{A}^{q|j}),\underline
{\mathcal{U}}) = \check{H}^{p}(M^{n|m}, \mathcal{H}^{q|j}).
\end{equation}
In particular $E^{p,q}_{2} = 0$ when $q > 0$, therefore $(E^{p,q}_{r}, d_{r})$
stabilizes at $r =2$. On the other hand we have
\begin{equation}
E^{^{\prime}p,q}_{2} = H_{DR}^{q} ( \check{H}^{p}( \mathcal{A}^{q|j}%
,\underline{\mathcal{U}})).
\end{equation}
We can easily see that the sheaves are fine i.e. that
\begin{equation}
\check{H}^{0}( \mathcal{A}^{q|j},\underline{\mathcal{U}}) = \mathbf{A}^{q|j}
\end{equation}
and
\begin{equation}
\check{H}^{p}( \mathcal{A}^{q|j},\underline{\mathcal{U}}) = 0 \ \ \text{when}
\ \ p>0.
\end{equation}
The latter identity can be proved using standard partitions of unity relative
to the covering $\underline{\mathcal{U}}$ of the underlying smooth manifold
$M$. Therefore we conclude that $(E^{^{\prime}p,q}_{r}, d_{r})$ also
stabilizes at $r=2$ and $E^{^{\prime}p,q}_{2} = 0$ when $p>0$ and
\begin{equation}
E^{^{\prime}0,q}_{2} = H_{DR}^{q|j}(M^{n|m}) .
\end{equation}
The theorem is then proved by using the fact that the two spectral sequences
must converge to the same thing and therefore
\begin{equation}
H_{DR}^{q|j}(M^{n|m}) = E^{^{\prime}0,q}_{2} \cong E^{q,0} = \check{H}^{q}(M^{n|m}, \mathcal{H}^{0|j}).
\end{equation}

\end{proof}

It may happen the sheaf $\mathcal{H}^{0|j}$ is actually a constant sheaf, for
instance on projective superspaces $\mathbb{P}^{n|m}$ the forms $\{\psi
_{i}\delta(d\psi_{i})\}$ are globally defined. In this case, as a corollary of
the above result, we obtain a sort of "Kunneth formula" for the super de Rham
cohomology on supermanifolds.

\begin{corollary}
Let $M^{n|m}$ be a super-manifold, such that $\mathcal{H}^{0|j}$ is a constant
sheaf, (e.g. when the locally defined forms $\{\psi_{i}\delta(d\psi_{i})\}$
extend globally). Then the de Rham cohomology of $M^{n|m}$ is:
\begin{equation}
H_{DR}^{*|j}(M^{n|m})=H_{DR}^{\ast}(M)\otimes\mathcal{H}^{0|j}.
\end{equation}

\end{corollary}

\begin{proof}
The map $\psi:H_{DR}^{\ast}(M)\otimes\mathcal{H}\longrightarrow H_{DR}^{\ast
}(M^{n|m})$ given by multiplication is a map in cohomology. It is easy to show
that, if $\gamma$ is an element of $H_{DR}^{\ast}(M)$ and $\omega$ is an
element of $\mathcal{H}$, then $\gamma\omega$ is an element of $H_{DR}^{\ast
}(M^{n|m})$. Moreover, if $\gamma$ and $\gamma^{\prime}$ are cohomologous in
$H_{DR}^{\ast}(M)$, then $\gamma\omega$ and $\gamma^{\prime}\omega$ are
cohomologous in $H_{DR}^{\ast}(M^{n|m})$: if $\gamma-\gamma^{\prime}=df$, then
$\gamma\omega-\gamma^{\prime}\omega=d(f\omega)$, since $d\omega=0$. Now, we
proceed by induction on the number of open sets of the good cover of $M$.
Obviously, if this number is equal to 1, then $M=\mathbb{R}^{n}$, and the
thesis is true for the study we have performed above. We have to prove the
truth of the thesis for an integer $s$, knowing that it is true for $s-1$. So,
let $M$ be covered by $s$ open sets forming a good cover. Then, we can call
$U$ one of them, and $V$ the union of the remaining ones. We know that the
thesis is true on $U$; $V$ and $U\cap V$. We will call $U^{m|n}$ and $V^{m|n}$
the open sets $U$ and $V$ endowed with the corresponding graded sheaves. Let
$k;p$ be two integers; by the usual Mayer-Vietoris sequence,%
\begin{equation}
...\longrightarrow H^{p}(U\cup V)\longrightarrow H^{p}(U)\oplus H^{p}%
(V)\longrightarrow H^{p}(U\cap V)\longrightarrow...
\end{equation}

If $\mathcal{H}^{q}$ are the elements of $\mathcal{H}$ of degree $\cdot|q$, we
have the following exact sequence:%
\begin{equation}
...\longrightarrow H^{p}(U\cup V)\otimes\mathcal{H}^{q}\longrightarrow
(H^{p}(U)\otimes\mathcal{H}^{q})\oplus(H^{p}(V)\otimes\mathcal{H}%
^{q})\longrightarrow(H^{p}(U\cap V)\otimes\mathcal{H}^{q})\longrightarrow...
\end{equation}

Summing up, we find that the following sequence is exact:%
\begin{align*}
... &  \longrightarrow\bigoplus_{p+q=k}H^{p}(U\cup V)\otimes\mathcal{H}^{q}\\
&  \longrightarrow\bigoplus_{p+q=k}(H^{p}(U)\otimes\mathcal{H}^{q}%
)\oplus(H^{p}(V)\otimes\mathcal{H}^{q})\\
&  \longrightarrow\bigoplus_{p+q=k}(H^{p}(U\cap V)\otimes\mathcal{H}%
^{q})\longrightarrow...
\end{align*}
where the sum is performed over $p,q$.

The following diagram is commutative:%
\begin{small}

\begin{align*}
\hspace{-1.5cm}\bigoplus_{p+q=k}H^{p}(U\cup V) &  \otimes\mathcal{H}%
^{q}\rightarrow\bigoplus_{p+q=k}(H^{p}(U)\otimes\mathcal{H}^{q}%
)\oplus(H^{p}(V)\otimes\mathcal{H}^{q})\rightarrow\bigoplus_{p+q=k}%
(H^{p}(U\cap V)\otimes\mathcal{H}^{q})\\
\, &  \downarrow\psi\,\qquad\qquad\qquad\qquad\qquad\qquad\downarrow\psi
\qquad\qquad\qquad\qquad\qquad\qquad\downarrow\psi\\
H^{k}( &  M^{n|m})\qquad\qquad\rightarrow\qquad H^{k}(U^{n|m})\oplus
H^{k}(V^{n|m})\qquad\rightarrow\qquad H^{k}((U\cap V)^{n|m})
\end{align*}

\end{small}

The commutativity is clear except possibly for the square:%
\begin{align*}
\oplus(H^{p}(U\cap V)\otimes\mathcal{H}^{q}) &  \longrightarrow^{d^{\ast}%
}\oplus H^{p+1}(U\cup V)\otimes\mathcal{H}^{q}\\
&  \downarrow\psi\text{ \ \ \ \ \ \ \ \ \ \ \ \ \ \ \ \ \ \ \ \ \ }%
\downarrow\psi\\
H^{k}((U\cap V)^{n|m})\qquad &  \longrightarrow^{d^{\ast}}\qquad\qquad
H^{k+1}(M^{n|m})
\end{align*}

Let $\omega\otimes\phi$ be in $(H^{p}(U\cap V)\otimes\mathcal{H}^{q})$. Then,
$\psi d^{\ast}(\omega\otimes\phi)=(d^{\ast}\omega)\cdot\phi$ and $d^{\ast}%
\psi(\omega\otimes\phi)=d^{\ast}(\omega\phi)$.

If $\{\rho_{U};\rho_{V}\}$ is a partition of unity subordinate to $\{U;V\}$,
then $d^{\ast}\omega=-d(\rho_{V}\omega)$ and $d^{\ast}(\omega\phi)=-d(\rho
_{V}\omega\phi)$ on $U$, while $d^{\ast}\omega=d(\rho_{U}\omega)$ and
$d^{\ast}(\omega\phi)=d(\rho_{U}\omega\phi)$ on $V$. Note that $-d(\rho
_{U}\omega\phi)=d(\rho_{V}\omega\phi)$ on $U\cap V$, since both $\omega$ and
$\phi$ are closed. So, $d^{\ast}(\omega\phi)$ is a global section of the sheaf
of $M^{n|m}$.

By these relations, it's easy to see that the square is commutative:

$d^{\ast}\psi(\omega\otimes\phi)=d^{\ast}(\omega\phi)=d(\rho_{U}\omega
\phi)=(d\rho_{U}\omega)\phi=(d^{\ast}\omega)\cdot\phi=\psi d^{\ast}%
(\omega\otimes\phi)$, since $\phi$ is closed.

By the Five Lemma, if the theorem is true for $U^{n|m}$, $V^{n|m}$ and $(U\cap
V)^{n|m}$ then it holds also for $M^{n|m}$, by induction.
\end{proof}

\end{document}